\documentclass[a4paper]{amsart}
\usepackage{amsmath,amssymb}
\usepackage{amsthm}
 
\usepackage{pifont}

\usepackage{bm,bbm}
\usepackage{xcolor}
\usepackage{graphicx}
\usepackage{wrapfig}
\usepackage{cutwin}

\usepackage{booktabs} 

\usepackage{hyperref}
\hypersetup{
    colorlinks=true,
    linkcolor=blue,
    urlcolor=blue,
    citecolor=blue!50!black
}

\graphicspath{{./figs/}}

\title{Tutte Embeddings of Tetrahedral Meshes
}
\author{Marc Alexa}

\thanks{Faculty of Computer Science \& Electrical Engineering, TU Berlin, \texttt{marc.alexa@tu-berlin.de}}
\address{TU Berlin, Sekretariat MAR 6-6, Marchstr. 6, 10587 Berlin, Germany}
\email{marc.alexa@tu-berlin.de}

\def\R{\mathbb{R}}
\def\tp{^{\mathsf{T}}}
\newcommand{\mv}[1]{\mathbf{#1}}

\newtheorem{theorem}{Theorem}
\newtheorem{lemma}{Lemma}[section]

\newtheorem{proposition}[lemma]{Proposition}
\newtheorem{corollary}[lemma]{Corollary}

\begin{document}

\maketitle

\begin{abstract}
Tutte's embedding theorem states that every 3-connected graph without a $K_5$ or $K_{3,3}$ minor (i.e.\ a \emph{planar} graph) is embedded in the plane if the outer face is in convex position and the interior vertices are convex combinations of their neighbors. We show that this result extends to simply connected tetrahedral meshes in a natural way: for the tetrahedral mesh to be embedded if the outer polyhedron is in convex position and the interior vertices are convex combination of their neighbors it is sufficient (but not necessary) that the graph of the tetrahedral mesh contains no $K_6$ and no $K_{3,3,1}$, and all triangles incident on three boundary vertices are boundary triangles.
\end{abstract}

\section{Introduction}

Every planar graph has a straight line embedding in the plane~\cite{Fary}. Tutte's celebrated embedding theorem~\cite{Tutte} provides a constructive proof for 3-connected planar graphs: if the outer polygon of the graph is in convex position and the interior vertices are convex combinations of their neighbors then the realization is an embedding (and every face is convex). This means, we can compute an embedding by placing the vertices of the boundary polygon so that it is convex and then solving a linear system for the positions of the interior vertices. This procedure has become an invaluable tool in computer graphics, geometry processing, and CAGD, where it is used to construct mappings for triangulations~\cite{Floater:2003:PLM}. It has been generalized to other surface geometries and topologies~\cite{GORTLER200683,Aigerman:2015:OTE,Aigerman:2016:HOT}. 

A particularly useful generalization of Tutte embeddings would be the three-dimensional case and, in particular, tetrahedral meshes. Embedding a given tetrahedral mesh with fixed convex boundary would enable bijective piecewise linear mappings between different domains represented by the same boundary triangulation. Alas, Tutte's result fails to generalize to 3D in a simple way for arbitrary polyhedral complexes~\cite{chilakamarri1995three}. It has been observed to fail even for well behaved tetrahedral meshes in practice~\cite[Fig.~2]{Campen}. We have conducted simple experiments on small tetrahedral meshes that suggests the convex combination approach will fail 
 for ``most'' randomly chosen convex combination weights. Note that in most practical case one starts with an embedded tetrahedral mesh. If the mesh is already realized with convex boundary, there clearly exists a convex combination map that would generate the realization (just express every interior vertex as the convex combination of its neighbors). Our experiments suggest that the chance of finding such weights by random sampling is small. 

The smallest possible counterexample is composed of two vertices inside a tetrahedron, with the two vertices connected to each other and all 4 boundary vertices~\cite{Floater:2006:CCM}. The graph of this tetrahedral mesh is the complete graph on 6 vertices  $K_6$. In view of the fact that Tutte's original work showed that 3-connected graphs may be embedded in the plane if (and only if) they have no $K_5$ and $K_{3,3}$ as a minor, one might ask if the observation that a $K_6$ is a counterexample for the tetrahedral case has any meaning. Note that $K_6$ is a forbidden minor in the class of \emph{linklessly embeddable} graphs, which have been termed the natural three-dimensional analogue of planar graphs~\cite{Sachs}.
An important difference to the planer case, of course, is that a tetrahedral mesh with vertices and edges forming a $K_6$ \emph{can} be embedded, meaning that there exist some convex combination weights that lead to an embedding (unlike in the case of non-planar graphs). We may still ask if tetrahedral meshes that are linklessly embeddable (i.e.\ have no $K_6$ and $K_{3,3,1}$ as a minor, see Section~\ref{sec:basics}) guarantee Tutte embeddings. Indeed, this turns out to be the case. 

Apart from excluded minors, however, there is one more obstruction. It has long been observed in the continuous domain that the harmonic extension from a boundary homeomorphism fails to be injective in 3D~\cite{Melas}. In fact, arbitrarily small perturbations of the identity map on the boundary are sufficient to cause the loss of injectivity~\cite{Laugesen}. For Tutte's embedding this has been elucidated by Floater\footnote{The quote is from private email exchange on the topic of tetrahedral Tutte embeddings} as follows:
\begin{quote}
[...] there are very many ways
to map, 1-1, one convex boundary into another
(imagine creating lots of twists and turns).
Won't that create foldover in the 3D embedding inside?
So I guess the boundary mapping needs to be restricted.
\end{quote}
We find that this problem materializes if an interior triangle in the tetrahedral mesh is fixed on the boundary, but the remaining vertices of its incident tetrahedra are not. Then "pulling" boundary vertices allows moving the interior vertices on either side of the fixed triangle. Interestingly, some proofs of Tutte's theorem require no chords in the graph. In 2D, this restriction can be lifted by observing that chords simply divide the outer polygon into two smaller polygons. In 3D, however, a triangle on boundary vertices is generally not dividing the mesh into two components. 

As we will show, the two restrictions mentioned above are sufficient to prove a version of Tutte's embedding theorem for tetrahedral meshes. The result is as follows:
\begin{theorem}
Given a tetrahedral mesh with the following properties:
\begin{itemize}
    \item The boundary is simply connected. 
    \item A triangle incident on three boundary vertices is on the boundary.
    \item The graph is 4-connected.
    \item The graph has no $K_6$ or $K_{3,3,1}$ as a minor.
\end{itemize}
If the vertex positions of the mesh are realized in $\R^3$ so that (1) the boundary triangles form a strictly convex polyhedron and (2) each interior vertex is a strictly convex combination of its neighbors then the mesh is embedded. 
\end{theorem}
We hasten to point out that this result has little direct consequence on the practice of using the commonly generated tetrahedral meshes for creating PL mappings: almost all of them have a $K_6$-minor and, consequently, a convex combination mapping will likely not be an embedding. We discuss possible practical consequences in Section~\ref{sec:discussion}.

\section{Tetrahedral meshes and Linkless embeddings}
\label{sec:basics} 

A tetrahedral mesh $T$ is a simplicial complex consisting of vertices, edges, triangular faces and tetrahedral cells. We assume the tetrahedral mesh is a \emph{topological ball}, meaning it has a simply connected interior. The boundary forms the graph of a polyhedron, a planar 3-connected graph with triangular faces~\cite{Steinitz}. Boundary faces are incident on one cell, interior faces are incident on two cells. In the following we assume that interior triangles satisfy the assumption of Theorem~1, i.e., they are incident on at most two boundary vertices. 

The \emph{star} of a vertex $v$ is formed by the simplices incident on $v$. The \emph{link} is the boundary of the star. We make the following observation about the connectivity of links for tetrahedral meshes satisfying assumptions of Theorem~1:
\begin{lemma}
The link $L_v$ of a vertex $v$ is 3-connected if interior triangles are not incident on only boundary vertices.
\label{lemma:link-connectivity}
\end{lemma}
\begin{proof}
If $v$ is an interior vertex, $L_v$ is 3-connected by Steinitz's theorem~\cite{Steinitz}.
Let $v$ be a boundary vertex and $V_b$ be the boundary vertices connected to $v$. The link $L_v$ is a planar triangulation with boundary $V_b$. Every edge in $L_v$ forms a triangle with $v$. A chord would induce a triangle that is not on the boundary of $T$, but all its vertices are boundary vertices. So there is no chord. This implies that the triangulation $L_v$ is 3-connected~\cite[3.2]{LAUMOND199087}. 
\end{proof}

Every interior vertex has degree at least 4. Boundary vertices with degree 3 create \emph{ears}: tetrahedra with 3 faces on the boundary connected to the remaining tetrahedral mesh through a single face. Since we ask that there are no interior faces incident on three boundary vertices, the tetrahedral meshes we consider contain no ears, and also boundary vertices have degree at least 4.


By \emph{tetrahedral graph} we mean the graph induced by the vertices and edges of $T$. 
Linkless embeddings of graphs, intuitively, are realizations of the graph in $\R^3$ so that there are no two cycles that are linked, i.e., that cannot be homotopically deformed so that the two cycles are topological disks. Sachs~\cite{Sachs} showed that the Petersen graphs (see Figure~\ref{fig:pg}) are intriniscally linked and minimal. He suspected that all linklessly embeddable graphs can be characterized as those without a Petersen graph as a minor. This was eventually proved by Robertson et al.~\cite{Robertson:1995}. 

Of the seven graphs in the Peterson family, we only need the complete graph $K_6$ and, to show that certain degeneracies cannot occur, the complete tripartite graph $K_{3,3,1}$. This seems quite natural, as the other graphs have vertices with degree less than 4, so cannot serve as minimal counter-examples for tetrahedral Tutte embeddings.

\begin{figure}
    \centering
    \includegraphics[width=0.5\linewidth]{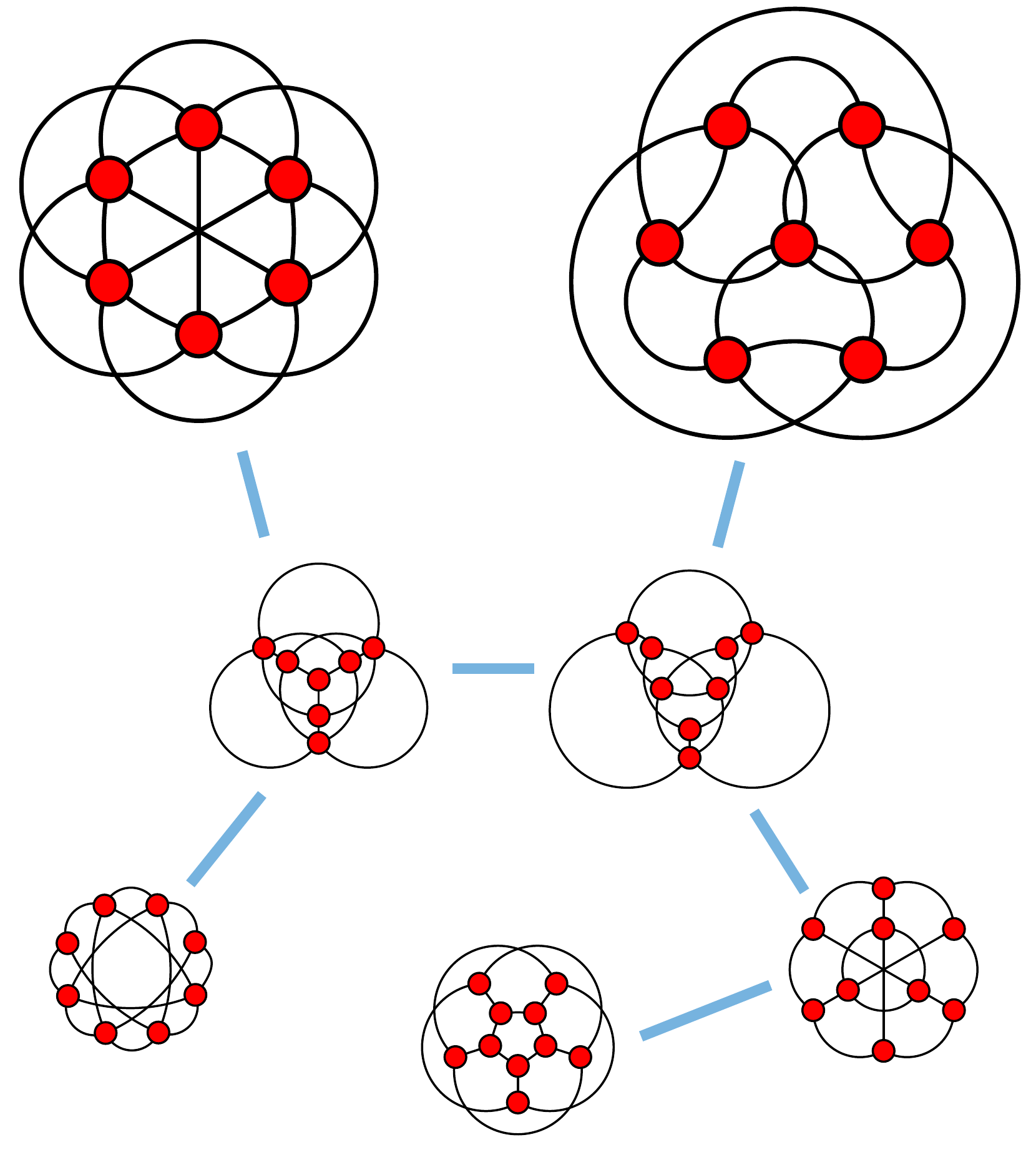}
    \caption{The seven graphs of the Petersen family. The larger ones in the top row, $K_6$ and $K_{3,3,1}$, are necessary for the characterization of tetrahedral meshes that can be embedded by convex combination mappings. A graph is linklessly embeddable if it has no minor in this family. Drawing adopted from David Epstein (public domain).}
    \label{fig:pg}
\end{figure}

The proof of the main statement is inspired by earlier proofs and their adaptations of Tutte~\cite{Tutte}, Geelen~\cite{Geelen}, Spielman~\cite{Spielman}, and the one focusing on PL mappings by Floater~\cite{Floater:2003:PLM}. 

\section{Realization by Convex Embedding, and Degeneracies}

The graph $T$ is \emph{realized} by assigning coordinates to the vertices, i.e. $V \mapsto \R^{3 \times |V|}$. We denote the coordinate of vertex $v$ as $\mv{x}(v) \in \R^3$. We assume the vertices on the boundary are realized so that the boundary faces form a convex polyhedron. 

Every interior vertex is realized as a convex combination of its neighbors:
\begin{equation}
    \mv{x}(v) = \sum_{(v,w) \in T} b_{(v,w)}\mv{x}(w),
    \quad
    \sum_{(v,w) \in T} b_{(v,w)} = 1,
    \quad
    b_{(v,w)} > 0.
    \label{eq:convex}
\end{equation}
This implies that $\mv{x}(v)$ lies strictly in the interior of the convex hull of its neighbors. In the following we want to show that (a) the convex hull of the neighbors cannot degenerate, i.e.\ be contained in a common plane; and (b) that $\mv{x}(v)$ lies strictly inside the boundary polyhedron for interior vertices $v$.

We start by recalling that a plane through $\mv{q}\in\R^3$ with normal vector $\mv{n} \in \R^3, \mv{n}\tp\mv{n} = 1$ is defined as the set
\begin{equation}
    P_{\mv{n},\mv{q}} = \{\mv{x} \in \R^3: \mv{n}\tp(\mv{x} - \mv{q}) = 0\}.
\end{equation}
The \emph{positive open half-space} is the set of points on the side of the plane in the direction $\mv{n}$, i.e.\ $P^+_{\mv{n},\mv{q}} = \{\mv{z} \in \R^3: \mv{n}\tp(\mv{z}-\mv{q}) > 0\}$. The negative open half-space is defined analogously. We make the following simple observation for interior vertices:
\begin{lemma}
Let $v$ be an interior vertex. Consider the plane $P_{\mv{n},\mv{x}(v)}$. If $v$ has a neighbor realized in $P^+_{\mv{n},\mv{x}(v)}$ then it has at least one neighbor in $P^-_{\mv{n},\mv{x}(v)}$ (and vice versa).
\label{lemma:both-halfspaces}
\end{lemma}
\begin{proof}
This follows by contradiction from $\mv{x}(v)$ being in the interior of the convex hull of its neighbors.
\end{proof}
If $v$ is a boundary vertex, it is possible that that all neighbors are contained in $P_{\mv{n},\mv{x}(v)}$ and either $P^+_{\mv{n},\mv{x}(v)}$ or $P^-_{\mv{n},\mv{x}(v)}$. In this case we call $v$ \emph{extreme} along $\mv{n}$, because the convexity of the boundary polyhedron implies that there is no other vertex $v'$ with $\mv{n}\tp(v') > \mv{n}\tp(v)$. We also observe that interior vertices are never extreme:

\begin{lemma}
An interior vertex is not realized on the boundary. 
\label{lem:strictly-inside}
\end{lemma}
\begin{proof}
Assume vertex $v_i$ is realized on a boundary element $b$ -- this may be a vertex, edge, or face. Let $V_i$ be the vertices realized in $b$ that are connected to $v_i$ by a path realized in $b$. All vertices in $V_i$ are realized in a common plane $P_{\mv{n},\mv{v}_i}$ and $\mv{n}$ can be chosen so that $\mv{n}\tp\mv{v}_i \ge \mv{n}\tp\mv{v}, v \in T$ because the boundary polyhedron is convex. We may delete the at most 3 boundary vertices incident on $b$ and $V_i$ remains connected to a vertex $v'$ outside $V_i$ because we assume $T$ to be 4-connected. The vertex $v'$ is not realized in $P_{\mv{n},\mv{v}_i}$ and $\mv{n}$, so by Lemma~\ref{lemma:both-halfspaces} there exists a vertex in $V_i$ with neighbors realized in both half-spaces $P^+_{\mv{n},\mv{x}(v)}$, which is a contradiction.  
\end{proof}

We say that a path $v_0,v_1,\ldots$ is \emph{non-decreasing} w.r.t.\ $\mv{n}$ if the realizations of the vertices along the path satisfy $\mv{n}\tp(v_0) \le \mv{n}\tp(v_1) \le \ldots$. We observe that any vertex is connected to an extreme boundary vertex by a non-decreasing path.
\begin{lemma}
    For any vertex $v_0$ and any direction $\mv{n}$ there is a non-decreasing path $v_0,v_1,\ldots,v_b$ satisfying $\mv{n}\tp(v_0) \le \mv{n}\tp(v_0) \le \ldots \le \mv{n}\tp(v_b)$, with $v_b$ on the boundary and extreme.
    \label{lemma:weakpath}
\end{lemma}
\begin{proof} 
Consider vertex $v_k$ along the path, which may be an interior or boundary vertex. If $v_k$ has a neighbor $v'$ satisfying $\mv{n}\tp\mv{x}(v') > \mv{n}\tp\mv{x}(v_k)$ set $v_{k+1} = v'$ and continue with $v_{k+1}$. 

Now assume we have found a vertex $v$ such that $\mv{n}\tp\mv{x}(v) \ge \mv{n}\tp\mv{x}(v')$ for all vertices $v'$ adjacent to $v$. If $v$ is on the boundary we are done. If $v$ is interior,  Lemma~\ref{lemma:both-halfspaces} forces $\mv{n}\tp\mv{x}(v) = \mv{n}\tp\mv{x}(v')$, meaning all neighbors of $v$ lie in the plane $P_{\mv{n},\mv{x}(v)}$. Let V be the set of vertices realized in $P_{\mv{n},\mv{x}(v)}$ that are connected to $v$ by a path realized in $P_{\mv{n},\mv{x}(v)}$. Because the boundary polyhedron is not flat there are vertices in $T$ not in $V$. But $T$ is connected so there must be a vertex $v' \in V$ connected to a vertex not in $P_{\mv{n},\mv{x}(v)}$. If $v'$ is interior, Lemma~\ref{lemma:both-halfspaces} implies that $v'$ has a neighbor $v''$ satisfying $\mv{n}\tp(v'') > \mv{n}\tp\mv{x}(v') = \mv{n}\tp\mv{x}(v_i)$. Because $v$ is connected to $v'$ through vertices contained in $P_{\mv{n},\mv{x}(v_i)}$ there is non-decreasing path through $V$ to $v''$, and the process continues. This shows that the process has to end in a boundary vertex $v_b$, which is extreme because it satisfies $\mv{n}\tp\mv{x}(v_b) \ge \mv{n}\tp\mv{x}(v')$ for all neighbors $v'$.
\end{proof}

\begin{corollary}
Let $P^+_{\mv{n},\mv{q}}$ be an open half-space and $V_{\mv{n},\mv{q}}^+$ be the vertices contained in it, i.e.\ $v \in V_{\mv{n},\mv{q}} \Rightarrow \mv{x}(v) \in P^+_{\mv{n},\mv{q}}$. The graph induced by $V_{\mv{n},\mv{q}}^+$ is connected.
\label{cor:connected}
\end{corollary}
\begin{proof}
Let $V^b_{\mv{n}}$ contain all boundary vertices $v_b$ that maximize $\mv{n}\tp\mv{x}(v_b)$. This set consists of either a single vertex, or two vertices incident on a common edge, or three vertices incident on a common face. Note that $V^b_{\mv{n}}$ is connected. Any vertex in $V_{\mv{n},\mv{q}}^+$ is connected to $V^b_{\mv{n}}$ by a non-decreasing path, which is entirely contained in $P^+_{\mv{n},\mv{q}}$. 
\end{proof}

We now show that the neighbors of a vertex cannot degenerate to a flat configuration if the graph is linklessly embeddable.  

\begin{lemma}
Let $v$ be a vertex in a tetrahedral graph $T$. If $v$ and all its neighbors are realized in a common plane $P_{\mv{n},\mv{v}}$ then $T$ has $K_{3,3,1}$ as a minor. 
\label{lemma:noflat}
\end{lemma}
\begin{proof}
Let $V$ be the set of vertices contained in $P$ that are connected to $v$ by a path realized in $P_{\mv{n},\mv{x}(v)}$.  We distinguish two types of vertices in $V$: vertices $V_i$ whose neighbors are all in $V$ (such as $v$); and vertices $V_b$, which have at least one neighbor not in $V$. Note that vertices in $V_i$ cannot be boundary vertices. 

Let $V^+$, resp.\ $V^-$ be the vertices contained in $P^{\pm}_{\mv{n},\mv{x}(v)}$. Both are connected (Corollary~\ref{cor:connected}). Since $v$ is interior, Lemma~\ref{lem:strictly-inside} implies that none of the vertices in $V_b$ is extreme along $\mv{n}$. So all vertices in $V_b$ are connected to vertices in both $V^+$ as well as $V^-$: for interior vertices this follows from Lemma~\ref{lemma:both-halfspaces} and for boundary vertices from the realization of the boundary as a convex polyhedron. 

Pick a vertex $v_b \in V_b$ that is connected to a vertex $v_i \in V_i$. Such pair must exist because V is connected. Consider the link $L_{v_b}$ of $v_b$. It contains $v_i$ as well as a vertices in $V^\pm$. The vertices in $V_i \cap L_{v_b}$ cannot be connected by an edge to the vertices in $(V^+ \cup V^-) \cap L_{v_b}$., so there must be a set of vertices $v_b^k \in $ in $V_b \cap L_{v_b}$ separating the two sets. This set contains at least 3 vertices, because $L_{v_b}$ is 3-connected (Lemma~\ref{lemma:link-connectivity}). 

Now notice that $V_i$, $V^+$, $V^-$ are connected to all $v_b^{1,2,3}$, and $v_b$ is connected to $V_i$, $V^+$, $V^-$ as well as $v_b^{1,2,3}$. This is a $K_{3,3,1}$ (see illustration in Figure~\ref{fig:noflat}).
\end{proof}

\begin{figure}
    \centering
    \includegraphics[width=0.35\linewidth]{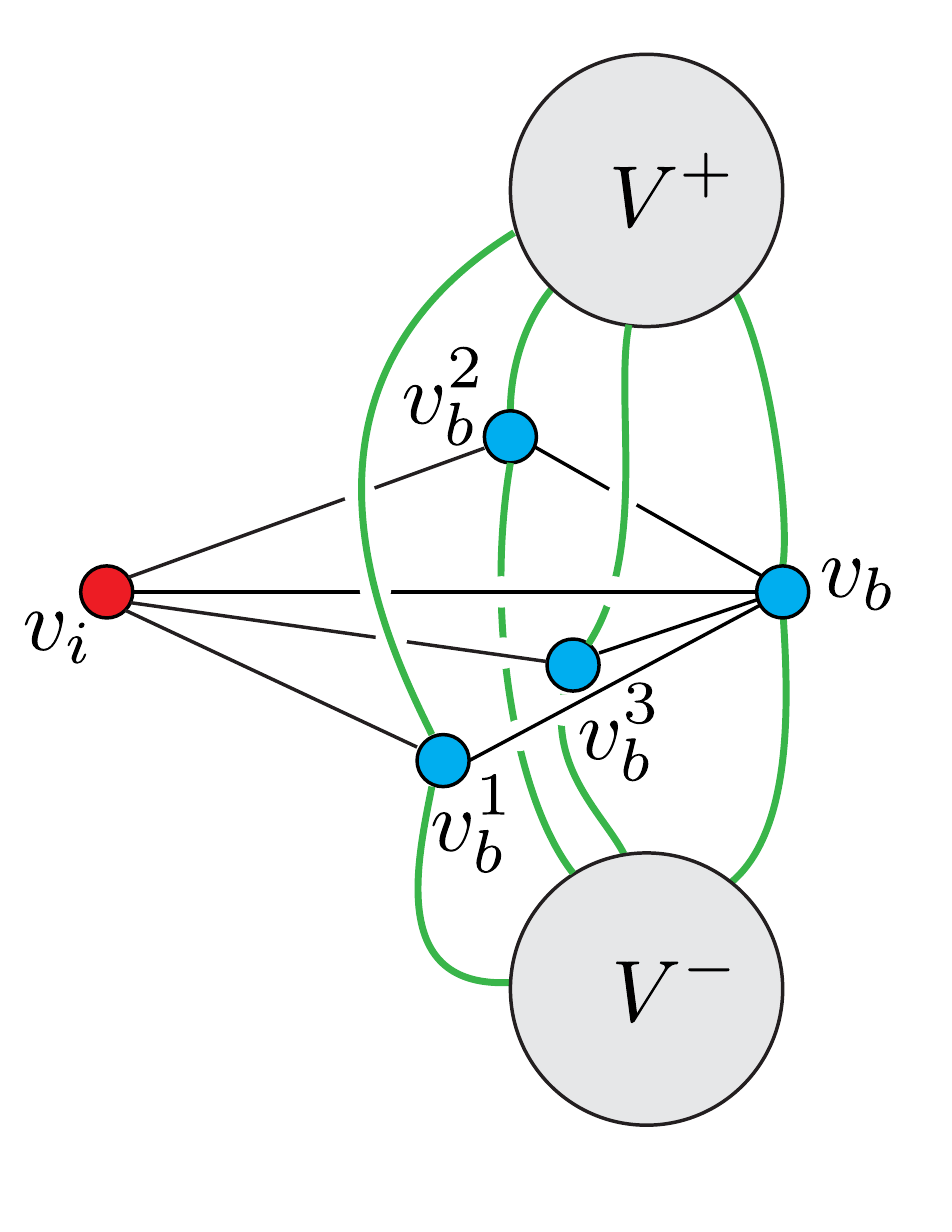}
    \caption{Illustration for Lemma~\ref{lemma:noflat}: Vertex $v_i$ and all its neighbors are in a common plane. Vertices $v_b^*$ are in the same plane but connected to vertices in both half-spaces. The vertices form a $K_{3,3,1}$.}
    \label{fig:noflat}
\end{figure}

By showing that the neighbors of a vertex are not co-planar, we have also shown that they are not co-linear or degenerate to a single point. Moreover, we have established that for any plane $P_{\mv{n},\mv{x}(v)}$, vertex $v$ will have neighbors in both half-spaces $P^{\pm}_{\mv{n},\mv{x}(v)}$. This observation allows us to sharpen Lemma~\ref{lemma:weakpath} to strictly increasing path $v_0, v_1, \ldots$ satisfying $\mv{n}\tp\mv{x}(v_i) < \mv{n}\tp\mv{x}(v_j)$ for $i < j$. 
\begin{corollary}
For any direction $\mv{n}$ every interior vertex $v$ has a strictly increasing path $v = v_0, v_1, \ldots, v_b$ that ends in a boundary vertex $v_b$. 
\label{cor:strictpath}
\end{corollary}
\begin{proof}
The proof is analogous to Lemma~\ref{lemma:weakpath} except that Lemmas~\ref{lemma:both-halfspaces} and \ref{lemma:noflat} now guarantee that every interior vertex $v_k$ has a neighboring vertex $v_{k+1}$ satisfying 
$\mv{n}\tp\mv{x}(v_{k+1}) > \mv{n}\tp\mv{x}(v_{k+1})$.
\end{proof}

\section{Local and global injectivity}

The fact that interior vertices are strictly in the interior of the boundary polyhedron (Lemma~\ref{lem:strictly-inside}) establishes the injectivity for tetrahedra incident on boundary triangles:
\begin{corollary}
Consider a boundary face $f$, its incident tetrahedron $t$ and the vertex $v$ in $t$ not in $f$. All elements of $t$ have positive (signed) measure.
\label{cor:nodegonboundary}
\end{corollary}
\begin{proof}
Since $v$ is strictly inside the boundary polyhedron, $t$ has positive volume. If any of the triangles or edges had zero area or lengths, the volume of the tetrahedron would be zero, so they are all strictly positive.
\end{proof}

Let us now consider two tetrahedra incident on an interior triangle. We want to show that if one of the tetrahedra is realized injectively then the other one must be as well. This local consistency we will then imply global injectivity. 

We need the following observation about \emph{planar triangulations}, i.e., a planar graph whose faces all have degree 3, except possibly the outer face. 

\begin{proposition}
In a planar triangulation without chords, the graph induced by the interior vertices is connected. Contracting the interior vertices  results in a triangulation of the boundary polygon with a single interior vertex connected to all boundary vertices. 
\label{prop:interior}
\end{proposition}
\begin{proof}
We assume the planar triangulation has at least two interior vertices. 
Assume that $a$ and $b$ are interior vertices not connected by a path of interior edges. Since there are no chords, the triangulation is 3-connected~\cite[3.2]{LAUMOND199087}, so $a$ and $b$ are connected to three boundary vertices $v_b^k, k = \{1,2,3\}$. Add vertex $e$ outside the boundary and connect it to $v_b^k$ with non-crossing paths. This creates a $K_{3,3}$, which is impossible, because the triangulation (including the paths to the exterior vertex) is planar. 

Since the interior vertices are connected they can be contracted to a single vertex without altering the boundary. Each boundary vertex is connected to a least one interior vertex, because there are no chords, so there are no ears. 
\end{proof}

%

\begin{figure}
    \centering
    \includegraphics[width = \linewidth]{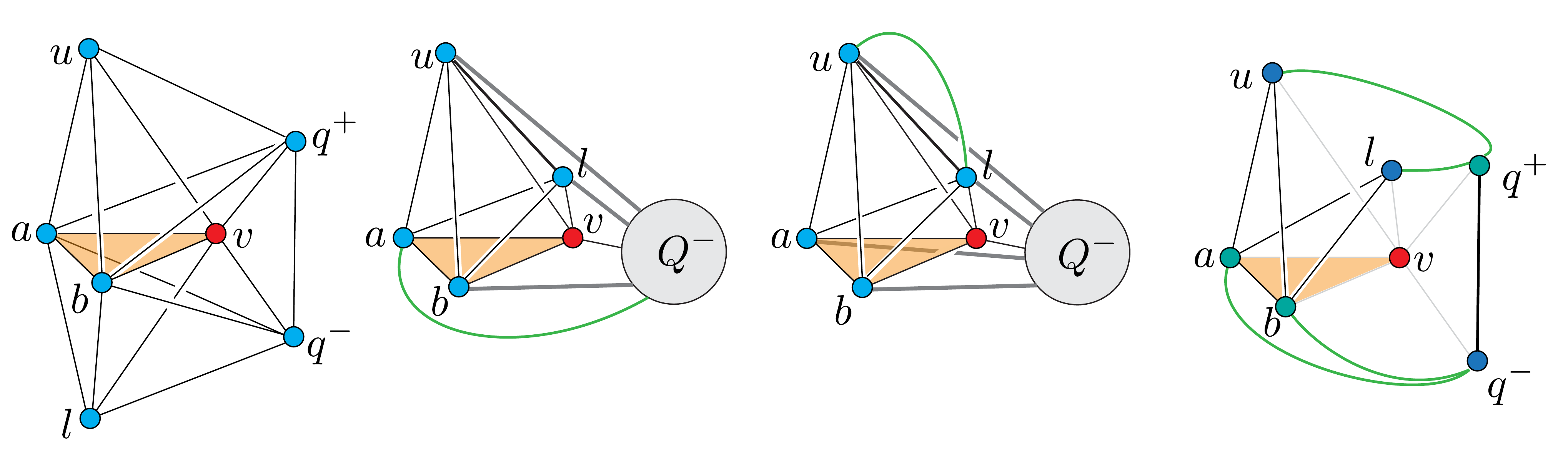}
    \caption{Illustration for Lemma~\ref{lemma:intface}: The interior face $f = (a,b,v)$ (in orange) is incident on two tetrahedra spanned by vertices $u$ and $l$. The leftmost illustration shows a proper realization. The other illustrations correspond the cases 1,2, and 3 (from left to right). }
    \label{fig:intface}
\end{figure}

\begin{lemma}
\label{lemma:intface}
Let $f$ be an interior triangle with incident tetrahedra $t_u,t_l$ and $u,l$ the vertices in $t_u,t_l$ not in $f$. Let $P_{\mv{n},\mv{x}(v)}$ be the plane through $f$ and assume $\mv{x}(u) \in P^+_{\mv{n},\mv{x}(v)}$. If $\mv{x}(l) \not\in P^-_{\mv{n},\mv{x}(v)}$ then $T$ contains $K_6$ or $K_{3,3,1}$ as a minor. 
\end{lemma}
\begin{proof}
Assume $\mv{x}(l)$ is contained $P_{\mv{n},\mv{x}(v)} \cup P_{\mv{n},\mv{x}(v)}^+$. As $f$ is an interior face it has at least one interior vertex. Let $v$ be the interior vertex of $f$ and $a,b$ the other two. The link $L_v$ of $v$ is a polyhedral graph. It contains the triangles $(a,b,u)$ and $(l,b,a)$. The remainder of $L_v$ is a triangulation bounded by the cycle $(a,u,b,l)$. For illustrations see Figure~\ref{fig:intface}.

\smallskip\noindent\emph{The vertex $q^-$}:
As $v$ has neighbors realized in $P_{\mv{n},\mv{x}(v)}^+$ it must have at least one neighbor in $P_{\mv{n},\mv{x}(v)}^-$ (by Lemma~\ref{lemma:both-halfspaces}). Every neighbor of $v$ is in $L_v$, but none of $a,u,b,l$ are realized in $P_{\mv{n},\mv{x}(v)}^-$. So there must be at least one other vertex in $L_v$. Among the vertices $q \in L_v$ realized in $P_{\mv{n},\mv{x}(v)}^-$ we pick one vertex and call it $q^-$. Note that all vertices $q \in L_v$ are connected to $v$.

\smallskip\noindent\emph{The path $C$}:
If $l$ is realized in $P^+_{\mv{n},\mv{x}(v)}$ there exists a path $C$ from $l$ to $u$ realized in $P^+_{\mv{n},\mv{x}(v)}$ by Lemma~\ref{cor:connected}. If $l$ is realized in $P_{\mv{n},\mv{x}(v)}$ it has a neighbor realized in $P^+_{\mv{n},\mv{x}(v)}$ by Lemmas~\ref{lemma:noflat} and \ref{lemma:both-halfspaces} (and because $f$ is not on the boundary). In this case there exists a path $C$ from $l$ to $u$ such that $C\setminus \{l\}$ is realized in $P_{\mv{n},c}^+$. We now consider three cases based on the intersection of $C$ and $L_v$.

\smallskip\noindent\emph{Case 1}: Vertices $u$ and $l$ are connected by an edge. This creates the two cycles $(a,l,u)$ and $(b,u,l)$ in $L_v$. The triangulation of one them contains $q^-$. Assume this is $b,u,l$ -- the other case is analogous. The triangulation of $(b,u,l)$ is connected by Prop.~\ref{prop:interior} and we contract it to $Q^-$, which is connected to $u,l$, and $b$. A path to $a$ is constructed from decreasing paths starting in $a$ and any vertex in $Q^-$. All vertices in this path, except for $a$ are realized in $P^-_{\mv{n},\mv{x}(v)}$, so the path cannot contain any of $b,u,l$, or $v$.
This means $a,u,b,l$ together with $v$ and $Q^-$ form a $K_6$.

\smallskip\noindent\emph{Case 2}: The edge $(u,l)$ is not present and $C \cap L_v = \{u,l\}$, i.e., the path $C$ has no intersection with $L_v$ except for the endpoints. In this case contract the interior vertices of the triangulation in the cycle $(a,u,b,l)$. The triangulation has no chords, because $(u,l)$ is not present, so it is connected (Prop.~\ref{prop:interior}) and we contract it into $Q^-$. All boundary vertices $a,u,b,l$ of the triangulation in the cycle are connected to $Q^-$ (Prop.~\ref{prop:interior}).
As above, $a,u,b,l$ together with $v$ and $Q^-$ form a $K_6$.

\smallskip\noindent\emph{Case 3}: The edge $(u,l)$ is not present and $C$ intersects $L_v$ in at least one vertex $q^+$. The name $q^+$ indicates that it is realized in  $P_{\mv{n},c}^+$ because $C\setminus l$ is. This implies that $q^+$ is distinct from $q^-$. Because $q^+ \in C$, it is connected to both $u$ and $l$. We construct decreasing paths from $a,b$, and $q^-$., establishing a path from $q^-$ to $a$ and to $b$. Each of $a$ and $b$ has a neighbor in $P_{\mv{n},\mv{x}(v)}$, so $a$ is not in the path from $q^-$ to $b$ and vice versa. 
Lastly, $q^+$ and $q^-$ are connected in $L_v\ \setminus\{a,u,b,l\}$, because the triangulation inside $(a,u,b,l)$ has no chords and is connected (Prop.~\ref{prop:interior}).
Now $\{u,l,q^-\}$ and $\{a,b,q^+\}$ form a $K_{3,3}$. Since all of the 6 vertices are in $L_v$, together with $v$ they form a $K_{3,3,1}$. 
\end{proof}

This establishes that if one of the two tetrahedra incident on an interior face is non-degenerate and correctly oriented, the other one is as well. Since the tetrahedra incident on the boundary are non-degenerate and correctly oriented (Corollary~\ref{cor:nodegonboundary}) and the dual graph of $T$ is connected it follows that all tetrahedra are non-degenerate and correctly oriented. 

The global injectivity (i.e.\ every point in the interior of the boundary is inside exactly one simplex) can be established by a homotopy argument (similar to the 2D case~\cite{Geelen,Spielman}): for any point $\mv{q}$ in the interior of the boundary consider a half-line $l$, originating at $\mv{q}$, not intersecting any vertex or edge. Such line exists because the shadows of the vertices and edges on the sphere of directions consists of finitely many at most one-dimensional subsets. Start with any point $\mv{x}$ on $l$ outside the boundary and move towards $\mv{q}$. The number of tetrahedra that contain $\mv{x}$ can only change if $\mv{x}$ crosses a face. As $\mv{x}$ crosses the boundary face, the number of tetrahedra containing $\mv{x}$ changes from $0$ to $1$. Note that no interior vertex, edge, or face can intersect any boundary face (by Corollary~\ref{cor:nodegonboundary}). In the interior, the number of tetrahedra containing $\mv{x}$ stays constant as $\mv{x}$ crosses interior faces by Lemma~\ref{lemma:intface}. Note that this is the case even if $l$ intersects more than one face in the same point. However, the fact that the number of tetrahedra containing any interior point is always one rules out the possibility that two faces intersect, because in the vicinity of the intersection more than one tetrahedron would contain the point. 

\section{Discussion}
\label{sec:discussion}

Whereas both $K_5$ and $K_{3,3}$ cannot be embedded in the plane, a tetrahedral mesh with $K_6$ as the underlying graph may well be embedded in $\R^3$. Whether a realization based on convex combinations is embedded depends on the weights. We have not been able to create a similar situation for a tetrahedral mesh with only a $K_{3,3,1}$-minor (and no $K_6$-minor). Or, in other words, we were unable to generate an example that would demonstrate the necessity of excluding $K_{3,3,1}$-minors. In fact, every graph without a $K_6$-minor turned out to be embedded by a any convex combination map. This leads to the question whether excluding $K_{3,3,1}$ is necessary for the proof.

There are certain bounds on the number $k$-cliques in a graph without $K_t$-minors~\cite{Wood:2016}. For the maximal number of 4-cliques in a graph on $n$ vertices without a $K_6$ minor we find $4n-15$. Since every tetrahedron is a 4-clique a tetrahedral mesh on $n$ vertices with more than $4n-15$ tetrahedra has to contain a $K_6$-minor. Consequently, only for tetrahedral meshes with few elements there is hope that convex combination mappings are guaranteed to work. This may be considered in contrast to the speculation by Chilakamarri et al.~\cite{chilakamarri1995three} that `high' connectivity is necessary for an extension of Tutte's theorem. If this is interpreted as the graph containing many edges, we find that rather the opposite is the case for tetrahedral meshes. 

Moreover, already the bound on the number of 4-cliques suggests that most tetrahedral meshes used in practice are guaranteed to contain a $K_6$ and, based on our and others' experiments, unlikely embedded by 'random' convex combination maps. One route for future investigations may be the reduction of the number of tetrahedra using bistellar flips, or the reduction of the graph by edge contractions.

\bibliographystyle{amsplain}
\bibliography{tutte3d}

\end{document}